\RequirePackage{fix-cm}
\documentclass[smallextended]{svjour3}       
\smartqed  
\usepackage{graphicx}
\begin{document}

\title{A 3/2--approximation for ``big'' two-bar charts packing\thanks{The work is supported by Mathematical Center in Akademgorodok under agreement No 075-2019-1613 with the Ministry of Science and Higher Education of the Russian Federation.}
}


\author{Adil Erzin  \and Stepan Nazarenko \and Gregory Melidi \and Roman Plotnikov
}


\institute{A. Erzin \at
              Sobolev Institute of Mathematics, Novosibirsk, Russia \\
              Tel.: +7-383-3297540\\
              Fax: +7-383-3332598\\
              \email{adilerzin@math.nsc.ru}           
}

\date{Received: date / Accepted: date}

\maketitle

\begin{abstract}
We consider a Two-Bar Charts Packing Problem (2-BCPP), in which it is necessary to pack two-bar charts (2-BCs) in a unit-height strip of minimum length. The problem is a generalization of the Bin Packing Problem (BPP). Earlier, we proposed an $O(n^2)$-time algorithm that constructs the packing which length at most $2\cdot OPT+1$, where $OPT$ is the minimum length of the packing of $n$ 2-BCs. In this paper, we propose an $O(n^4)$-time 3/2-approximate algorithm when each BC has at least one bar greater than 1/2. 

\end{abstract}

\section{Introduction}
\label{intro}
We faced the need to solve the problem when determined the start times for oil and gas field development projects in compliance with the annual production limits \cite{Erzin20}.

Mathematically the problem is as follows. Let us have a semi-infinite unit-height horizontal strip and a set of bar charts consisting of two bars, each with a height of at most $1$ and unit length. For convenience, BC, with $b$ bars, we will denote by $b$-BC. All 2-BCs are required to pack in a strip of minimum length. When packing BCs, crossing bars are naturally prohibited. Moreover, the bars of each BC can move vertically, but they are inseparable horizontally and cannot be interchanged. An example of feasible packing is in Fig. 1.

\begin{figure*}
\includegraphics[width=\textwidth]{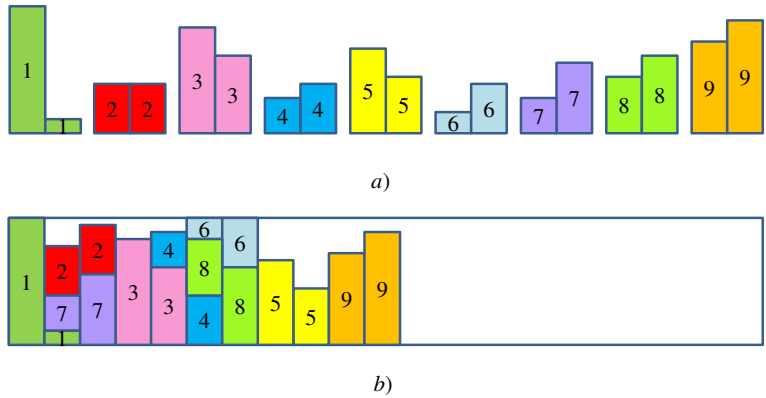}
\caption{The example of feasible packing. \emph{a}) Set of 2-BCs; \emph{b}) Packing of length 11.}
\label{fig1}
\end{figure*}

A similar well-studied problem is the Bin Packing Problem (BPP) \cite{Baker85,Dosa07,Johnson73,Johnson85,Li97,Yue91,Yue95}. In the BPP, we have a set of items $L$ with the sizes not more than $1$ and a set of unit-size containers (bins). It is required to put all items in a minimum number of bins. One packing algorithm is First Fit Decreasing (FFD). All items in the decreasing order placed in the first suitable bin.  Johnson proved that the FFD algorithm uses no more than $11/9\ OPT(L)+4$ bins \cite{Johnson73}. Backer reduced the additive constant to 3 \cite{Baker85}. Yue proved that $FFD(L)\leq 11/9\ OPT (L)+1$ \cite{Yue91}. Then together with Li \cite{Li97}, he improved the result to $FFD(L)\leq 11/9\ OPT (L)+7/9$. D\'{o}sa found the tight boundary of the additive constant and gave an example when $FFD(L)=11/9\ OPT(L)+6/9$ \cite{Dosa07}. A Modified First Fit Decreasing (MFFD) algorithm improves FFD by dividing items into groups by size and packing items from different groups separately. Johnson and Garey proposed this modification and showed that $MFFD(L)\leq 71/60\ OPT(L)+31/6$ \cite{Johnson85}. Subsequently, the result was improved by Yue and Zhang to $MFFD(L)\leq 71/60\ OPT(L)+1$ \cite{Yue95}.

The problem under consideration is also a particular case of the project scheduling problem when each job during a one-time slot consumes the limited non-accumulative resource \cite{Hartmann02,Kolisch06}. For the case of an accumulative resource, an exact algorithm has been developed \cite{Gimadi03}. In the case of a limited renewable resource, the problem is NP-hard, and polynomial algorithms with guaranteed accuracy estimates are not known \cite{Goncharov14,Goncharov17,Hartmann02,Kolisch06}.

The rest of the paper is organized as follows. Section 2 provides a statement of the packing problem for 2-BCs as a Boolean Linear Programming (BLP). Section 3 describes the $O(n^{3.5})$-time algorithm $M$ for packing of 2-BCs, which in the case when the 2-BCs are non-increasing, and the first bar is more than 1/2 yields a 3/2-approximate solution. Section 4 presents the $O(n^4)$-time algorithm $M_w$, which is in the case when at least one bar is more than 1/2 builds a 3/2-approximate solution. In section 5, we summarize and outline directions for further research.

\section{Formulation of the problem}
On the plane, we have a unit-height semi-infinite horizontal strip and a set of 2-BCs $S$ ($|S|=n$). Each 2-BC $i\in S$ consists of two unit-length bars. The height of the first bar is $a_i\in (0,1]$ and of the second $b_i\in (0,1]$. Let us divide the strip into equal rectangles (cells) of unit length and height and renumber them starting from the origin of the strip with integers $1,2,\ldots$.

\begin{definition}
2-BC $i$ is \emph{non-increasing} (\emph{non-decreasing}) if $a_i\geq b_i$ ($a_i\leq b_i$).
\end{definition}

\begin{definition}
\emph{Packing} is a function $p:S\rightarrow Z^+$, which associates with each BC $i$ the cell number of the strip $p(i)$ into which the first bar of BC $i$ falls.
\end{definition}

As a result of packing $p$, bars from 2-BC $i$ occupy the cells $p(i)$ and $p(i)+1$.

\begin{definition}
The packing is \emph{feasible} if the sum of the bar's heights that fall into one cell does not exceed 1.
\end{definition}

\begin{definition}
The packing \emph{length} $L(p)$ is the number of strip cells in which falls at least one bar.
\end{definition}

We assume that any packing $p$ begins from the first cell, and in each cell from 1 to $L(p$), there is at least one bar. If this is not the case, then all or part of the packing can be moved to the left.

The BLP formulation for 2-BCPP made in \cite{Erzin20_2}. Since \cite{Erzin20_2} is still arXiv paper, we repeat the BLP for convenience here. To do this, we introduce the variables:\\
$$
x_{ij}=\left\{
            \begin{array}{ll}
              1, & \hbox{if the first bar of BC $i$ is in the cell $j$;} \\
              0, & \hbox{else.}
            \end{array}
          \right.
$$
$$
y_j=\left\{
            \begin{array}{ll}
              1, & \hbox{if the cell $j$ contains at least one bar;} \\
              0, & \hbox{else.}
            \end{array}
          \right.
$$
Then 2-BCPP is written as follows.
\begin{equation}\label{e1}
  \sum\limits_j y_j \rightarrow\min\limits_{x_{ij},y_j\in\{0,1\}};
\end{equation}
\begin{equation}\label{e2}
  \sum\limits_j x_{ij} =1,\ i\in S;
\end{equation}
\begin{equation}\label{e3}
  \sum\limits_i a_ix_{ij} + \sum\limits_k b_kx_{k,j-1}\leq y_j,\ \forall j.
\end{equation}

The 2-BCPP is strongly NP-hard as the generalizations of the BPP \cite{Johnson73}. Moreover, the problem is $(3/2-\varepsilon)$-inapproximable unless P=NP \cite{Vazirani01}.

In \cite{Erzin20_2}, we proposed an $O(n^2)$-time algorithm, which packs the 2-BCs in the strip of length at most $2\cdot OPT+1$, where $OPT$ is the minimum packing length. In this paper, we propose two new $O(n^{3.5})$- and $O(n^4)$-time packing algorithms based on the sequential matching and prove that if at least one bar of each BC has a height greater than 1/2, then the constructed solution is 3/2-approximate. We show that this is a tight estimation.

\section{Algorithm $M$}
\begin{definition}
  If two BCs share three (two) cells, then we call this situation 1-union (2-union).
\end{definition}

Using the set $S$, we construct a graph $G_1=(V_1,E_1)$ in which the vertices are the images of BCs ($|V_1|=|S|=n$), and an edge connects two BCs if they can create ether 1- or 2-union.

Algorithm $M$ consists of a sequence of the steps. At the first step, in the graph $G_1$, the maximum matching of cardinality $m_1$ is constructed. The result are $n-m_1$ 2- and 3-BCs, which are the prototypes of vertices forming the set $V_2$ of the new graph $G_2=(V_2,E_2)$. The edge $(i,j)\in E_2$ if BCs $i$ and $j$ form a union. At an arbitrary step in the corresponding graph $G_k$, we construct the next maximum matching of cardinality $m_k$. The algorithm stops when in the graph $G_{p+1}$, there are no more pairs of BCs to combine. In Fig. 2, we illustrated the operation of the algorithm.

\begin{figure*}
\includegraphics[width=\textwidth]{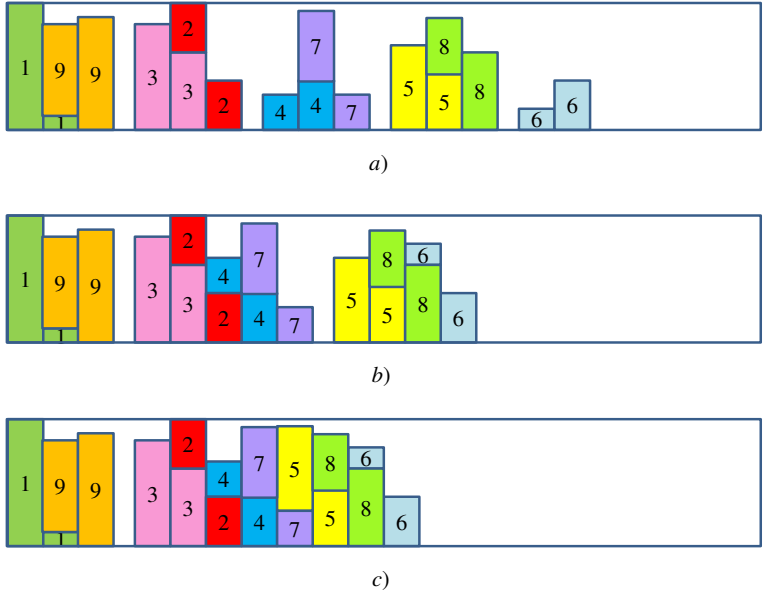}
\caption{Illustration of algorithm $M$ operation. \emph{a}) First matching, $m_1=4$; \emph{b}) Second matching, $m_2=2$; \emph{b}) Third matching, $m_3=1$.}
\label{fig2}
\end{figure*}

The length of the packing constructed by the algorithm $M$ is
\begin{equation}\label{e4}
  L_M(n)=2n-m_1-m_2-\ldots - m_p,\ p\geq 1.
\end{equation}

\begin{definition}
  If at least one bar in BC has height more than 1/2, then such BC we call \emph{big}.
\end{definition}

Let each BC is big and non-increasing. Then in each cell can be no more than two bars, and two BCs can form only a 1-union.

\begin{figure*}
\includegraphics[width=\textwidth]{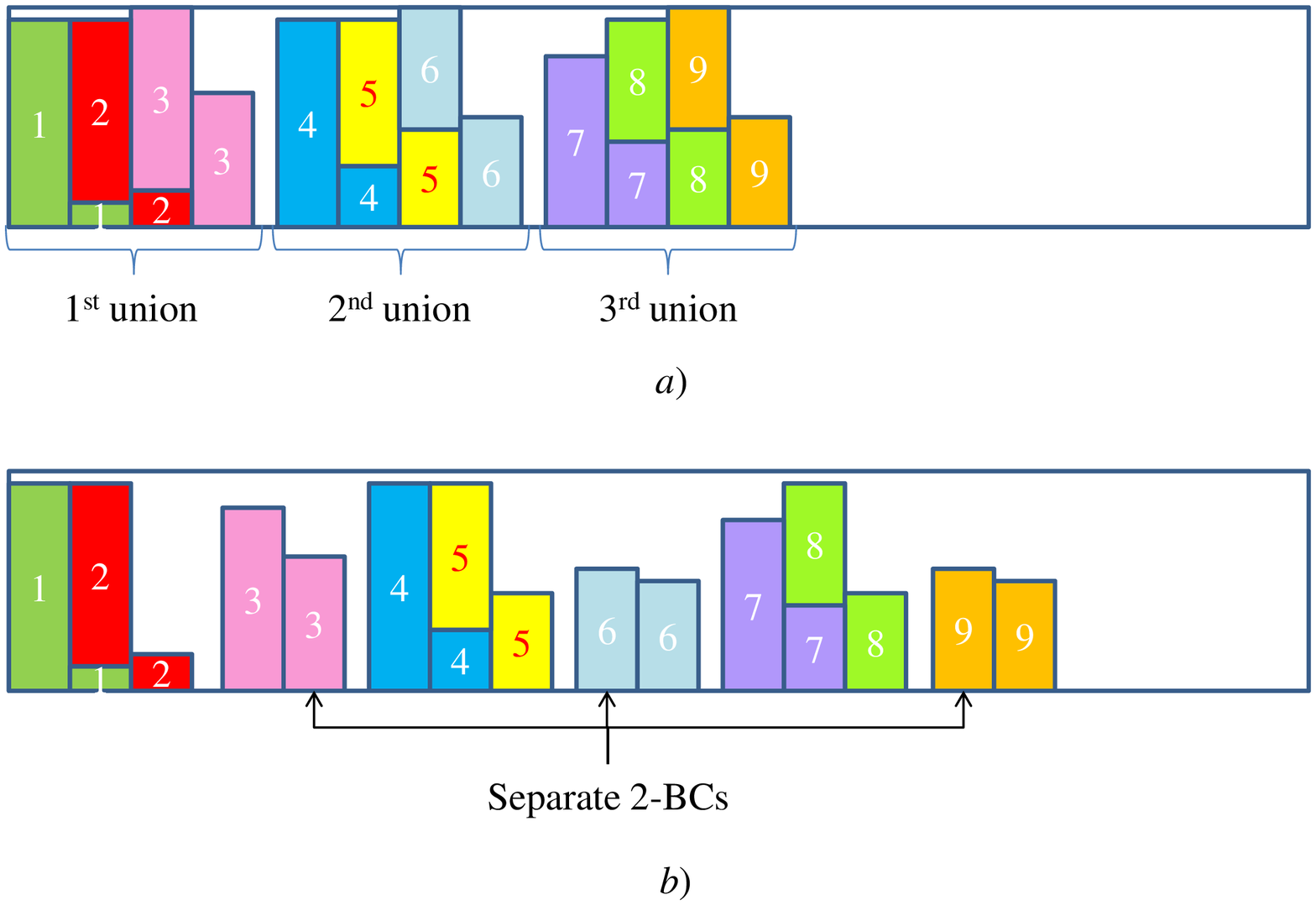}
\caption{Disassembling of optimal packing. \emph{a}) Optimal packing; \emph{b}) First optimal matching.}
\label{fig3}
\end{figure*}

Any packing, including the optimal one, can be disassembled into BCs obtained from the first matching. Let there be optimal packing. Let us single out independent unions, the beginning, and end of which is determined by the presence of one bar in the cell (there are 3 of them in the example in Fig. 3). Each union we will disassemble independently from left to right, separating 3-BCs consisting of two 2-BCs. If the union has an odd number of 2-BCs, then the last one will be one 2-BC (in the example in Figure 3, these are BCs 3, 6, and 9 in different unions).

Thus, the procedure for constructing optimal packing can be represented as a process of sequential construction of maximum matchings. Denote by $m_k^*$ the cardinality of the $k$th matching in the optimal packing. Then the length of optimal packing is
\begin{equation}\label{e5}
  OPT=2n-m_1^*-m_2^*-\ldots -m_q^*,\ q\geq 1.
\end{equation}

\begin{lemma}
 \begin{equation}\label{e6}
   m_2^*+\ldots +m_q^*\leq m_1^*\leq m_1.
 \end{equation}
\end{lemma}
\begin{proof}
We will need the following obvious
\begin{property}
 If there are $X$ BCs occupying a total of $Y$ cells, then after packing them, they will occupy at least $Y-(X-1)$ cells, i.e., packing length will decrease by no more than $X-1$.
\end{property}
Suppose that the optimal packing consists of $U\geq 1$ unions (which means $OPT=n+U$), and after disassembling, there are $B\in [0,U]$ separate 2-BCs. In Fig. 3 for example $U=3$, $B=3$, and $m_1^*=3$. $B$ 2-BCs cannot unite with each other. Otherwise, the disassembled packing is not optimal. Each of them is the last 2-BCs in each union. Total number of BCs after disassembling is $B+(n-B)/2=B+m_1^*$ since $(n-B)/2=m_1^*$. The first matching effects reduction in the number of occupied cells at most $B$ cells due to the union of $B$ 2-BCs. There are remain another $n-B$ 3-BCs, which can be combined only in their unions. Therefore, it is still possible to reduce the packing length by a maximum of $(n-B)/2-U$ cells (Property 1). Total, after the first matching, we can still reduce the packing length by
$$
 m_2^*+\ldots +m_q^*\leq B+(n-B)/2-U=B+m_1^*-U\leq m_1^*\leq m_1,
$$
since $B\leq U$ and by construction, $m_1 \geq m_1^*$. The lemma is proved.
\end{proof}

\begin{theorem}
 If all 2-BCs are big and non-increasing, then the algorithm $M$ constructs a 3/2-approximate solution for the 2-BCPP with time complexity $O(n^{3.5})$.
\end{theorem}
\begin{proof}
From (\ref{e4}), (\ref{e5}), (\ref{e6}) and $m_1\leq n/2$ it follows that $L=L_M(n)\leq 2n-m_1$, $OPT\geq 2n-2m_1$ and hence
$$
\frac{L}{OPT}\leq \frac{2n-m_1}{2n-2m_1}=1+\frac{m_1}{2n-2m_1}\leq 1+\frac{n/2}{2n-n}=\frac{3}{2}.
$$
The algorithm, $M$, uses the procedure for constructing the maximum matching at most $O(n)$ times. In \cite{Xie18}, an $O(n^{2.5})$-time algorithm proposed for constructing the maximum matching. So the time complexity of the proposed algorithm is $O(n^{3.5})$.
\end{proof}

Naturally, in the case of non-decreasing big 2-BCs, the algorithm $M$ also constructs a 3/2-approximate solution.

In \cite{Vazirani01}, Theorem 9.2 states that for any $\varepsilon > 0$, there is no approximation algorithm having a guarantee of $3/2-\varepsilon$ for the bin packing problem, assuming $P\neq NP$. Since the 2-BCPP is the generalization of BPP, this is the case for 2-BCPP too. However, if all BCs are big, we need to prove the tightness of the estimate, and we will do it in the next section.

\section{Algorithm $M_w$}
Now let all BCs are big, but not necessarily non-increasing or non-decreasing. There are more union options, and we will distinguish 1-unions and 2-unions (refer to Definition 5). We will construct now a \emph{weighted} graph $G_1=(V_1, E_1)$, in which, as before, the vertices are images of BCs. An edge between the vertices exists if these vertices can create a union. The weight of the edge $(i,j)\in E_1$ equals 1 if BCs $i$ and $j$  create a 1-union and 2 if they form a 2-union. In the algorithm, $M_w$, instead of the maximum matching at each step, a \emph{max-weight matching} constructed. There are no more differences from the algorithm $M$. We introduce the following additional notation:
\begin{itemize}
  \item $w_k^*$ is the weight of the $k$th matching in the optimal packing;
  \item $w_k$ is the weight of the $k$th matching in the packing constructed by algorithm $M_w$;
  \item $k_1^*$ is the number of 2-unions in the first matching in the optimal packing;
  \item $k_1$ is the number of 2-unions in the first matching of maximum weight constructed by algorithm $M_w$.
\end{itemize}
Since each BC with the above properties can participate in the 2-union only once, the following property is valid.
\begin{property}
 2-unions can only be when constructing the first matching.
\end{property}
Then $w_k^*=m_k^*,\ k\geq 2$. One can see the example of disassembling of optimal packing in Fig. 4.

\begin{figure*}
\includegraphics[width=\textwidth]{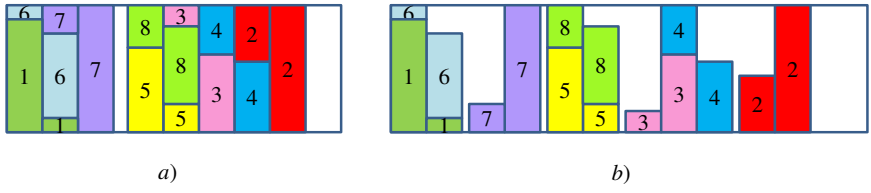}
\caption{Disassembling of optimal packing. \emph{a}) Optimal packing; \emph{b}) First optimal matching.}
\label{fig4}
\end{figure*}

\begin{lemma}
 The minimal length of packing $n$ big BCs is at least $n$ ($OPT\geq n$).
\end{lemma}
\begin{proof}
Suppose that the optimal packing contains $k$ 2-unions (which, according to Property 2, can no longer be combined). They reduce packing length by $2k$ cells. The remaining $n-2k$ BCs can only be combined into 1-unions, either with each other or with 2-unions, reducing the packing length by a maximum of $n-2k$ cells (Property 1). Therefore, $OPT\geq 2n-2k-(n-2k)=n$.
\end{proof}

\begin{theorem}
  If all 2-BCs are big, then with time complexity of $O(n^4)$, the algorithm $M_w$ constructs a 3/2-approximate solution to the 2-BCPP.
\end{theorem}
\begin{proof}
The first matching of maximum weight contains $k_1$ 2-unions. In the optimal packing, after disassembling, the first matching contains $k_1^*$ 2-unions. Because of Lemma 1, we have:
$$
 OPT=2n-w_1^*-m_1^*-\ldots -m_q^*\geq 2n-w_1^*-m_1^*.
$$
Therefore, taking into account inequality $L=L_{M_w}(n)\leq 2n-w_1$, we have
$$
 \varepsilon=\frac{L}{OPT}\leq\frac{2n-w_1}{2n-w_1^*-m_1^*}\leq\frac{2n-w_1}{2n-w_1-n/2}=1+\frac{1}{3-2w_1/n}=f(x),
$$
where $x=w_1/n$.

On the other hand (by Lemma 2) $OPT\geq n$. Hence
$$
 \varepsilon=\frac{L}{OPT}\leq\frac{2n-w_1}{n}=2-w_1/n=g(x).
$$
Therefore, $\varepsilon\leq\min\{f(x),g(x)\}$. The function $f(x)$ is increasing, and $g(x)$ is decreasing. Let $f(x_0)=g(x_0)$. Then $\varepsilon\leq f(x_0)=g(x_0)$. To find $x_0$, solve the equation $1+\frac{1}{3-2x}=2-x$, or $2x^2-5x+2=0$. There is one suitable solution $x=1/2$. Then $\varepsilon\leq g(1/2)=3/2$.

The time complexity of constructing a max-weight matching is $O(n^3)$ \cite{Gabow90,Galil86}, then the complexity of algorithm $M_w$ is $O(n^4)$.
\end{proof}

In Fig. 5, we give an example of the asymptotical attainability of the obtained estimate. In this example, $4k$ big 2-BCs: $2k$ of them are green non-increasing, and $2k$ are red non-decreasing. In the optimal packing, we construct the first matching shown in Fig. 5\emph{a}. Then the optimal packing is in Fig. 5\emph{b}, and $OPT=4k+1$. Algorithm $M_w$ can construct first max-weight matching, as in Fig. 5\emph{c}, and then there is no any more matching, and the length of the packing is $L=L_{M_w}=8k-2k$. Therefore, when $k$ tends to $\infty$, $L/OPT=6k/(4k+1)$ tends to 3/2.
\begin{figure*}
\includegraphics[width=\textwidth]{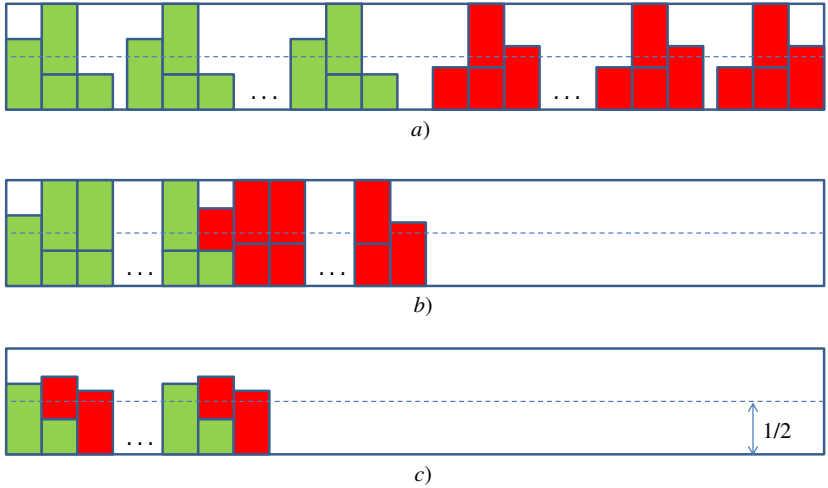}
\caption{\emph{a}) 3-BCs after the first optimal matching; \emph{b}) Optimal packing; \emph{b}) First max-weight matching built be $M_w$.}
\label{fig5}
\end{figure*}

\section{Conclusion}
We considered a new problem, in which it is necessary to pack two-bar charts in a unit-height strip of minimum length. The problem is a generalization of the bin packing problem. Earlier, we proposed an $O(n^2)$-time algorithm, which builds packing of length at most $2\cdot OPT+1$ for the general case. In this paper, we proposed an $O(n^4)$-time 3/2-approximate algorithm for the case when each BC has at least one bar greater than 1/2.

In future research, we plan to get a new estimate not only for big but also for arbitrary BCs.

\end{document}